\documentclass[11pt, oneside]{article}   	
\usepackage{geometry}                		
\usepackage{graphicx}				
\usepackage{tikz}

\usepackage{titling}
\usepackage{float}

\usepackage[pdftex,colorlinks=true,linkcolor=blue,citecolor=blue,urlcolor=black]{hyperref}						
						
\usepackage{thmtools,thm-restate}								
\usepackage{amssymb}
\usepackage{amsmath}
\usepackage{amsthm}
\usepackage{braket}

\usepackage{authblk}

\usepackage{thmtools,thm-restate}		
\newtheorem{theorem}{Theorem}
\newtheorem{lemma}{Lemma}
\newtheorem{corollary}{Corollary}
\newtheorem{definition}{Definition}
\newtheorem{claim}{Claim}

\DeclareMathOperator{\E}{\mathbb{E}}

\newcommand{\PR}{\mathsf{PostR}}

\newcommand{\PostBQP}{\mathsf{PostBQP}}
\newcommand{\PP}{\mathsf{PP}}
\newcommand{\Q}{\mathsf{Q}_2}
\newcommand{\C}{\mathsf{C}}
\newcommand{\D}{\mathsf{D}}
\newcommand{\N}{\mathsf{N}}
\newcommand{\NQ}{\mathsf{NQ}}
\newcommand{\PQ}{\mathsf{PostQ}}

\newcommand{\ndeg}{\mathsf{ndeg}}
\newcommand{\PRe}{\PR_\epsilon}
\newcommand{\PRR}{\PR_{1/3}}
\newcommand{\A}{\mathcal{A}}

\newcommand{\rdeg}{\mathsf{rdeg}}
\newcommand{\PostBPP}{\mathsf{PostBPP}}

\renewcommand{\deg}{\mathsf{deg}}

\newcommand{\EE}{\mathop{\E}}
\newcommand{\n}{\frac{N}{2}}
\newcommand{\stsum}{\sum_{\substack{S,T \subseteq [N]\\ 1 \leq |S|+|T| \leq d}}}

\title{Post-selected classical query complexity}
\author[]{Chris Cade\footnote{\nolinkurl{chris.cade@bristol.ac.uk}}}
\affil[]{School of Mathematics, University of Bristol, UK}
				
\begin{document}
\maketitle

\begin{abstract}
We study classical query algorithms with post-selection, and find that they are closely connected to rational functions with nonnegative coefficients. We show that the post-selected classical query complexity of a Boolean function is equal to the minimal degree of a rational function with nonnegative coefficients that approximates it (up to a factor of two). For post-selected \emph{quantum} query algorithms, a similar relationship was shown by Mahadev and de Wolf~\cite{mahadev2015rational}, where the rational approximations are allowed to have negative coefficients. Using our characterisation, we find an exponentially large separation between post-selected classical query complexity and post-selected quantum query complexity, by proving a lower bound on the degree of rational approximations (with nonnegative coefficients) to the Majority function. This lower bound can be generalised to arbitrary symmetric functions, and allows us to find an unbounded separation between non-deterministic quantum and post-selected classical query complexity. All lower bounds carry over into the communication complexity setting. 

We show that the zero-error variants of post-selected query algorithms are equivalent to non-deterministic classical query algorithms, which in turn are characterised by nonnegative polynomials, and for some problems require exponentially more queries to the input than their bounded-error counterparts. Finally, we describe a post-selected query algorithm for approximating the Majority function, and an efficient query algorithm for approximate counting. 
\end{abstract}

\section{Introduction}\label{sec:intro}
Post-selection is the (hypothetical) power to discard all the runs of a computation in which a particular event does not occur. For instance, suppose that we have a Boolean formula, and we wish to find an assignment of values to the variables that makes the formula evaluate to true. If we are guaranteed that such an assignment exists, then post-selection allows us to easily find it: we just guess a random assignment of values to the variables, and then post-select on the event that the formula evaluates to true. Thus, an ordinarily difficult problem becomes trivial when we have the power to post-select. We might then ask the following question: given the power of post-selection, which problems become easy to solve, and which remain difficult? 

The complexity class that captures this model of computation is $\PostBPP$~\cite{han1997threshold,bpppath}: bounded-error probabilistic polynomial time ($\mathsf{BPP}$) with post-selection. More precisely, $\PostBPP$ is the class of languages $L$ for which there exist two polynomial-time randomised classical algorithms $\mathcal{A}$ and $\mathcal{B}$ such that for all $x$,
\begin{eqnarray*}
\text{If }x\in L \qquad \Pr[\mathcal{A}(x) = 1 | \mathcal{B}(x) = 1] \geq 2/3 \\
\text{If }x\notin L \qquad \Pr[\mathcal{A}(x) = 1 | \mathcal{B}(x) = 1] \leq 1/3 
\end{eqnarray*}
and $\Pr[\mathcal{B}(x) = 1] > 0$, where the probabilities are taken over the random bits given to the algorithm. We can think of the algorithm $\mathcal{B}$ as the `post-selector', which tells us whether or not to accept the output of $\mathcal{A}$. An alternative characterisation of the class is in terms of computational path lengths: here, in contrast to the class $\mathsf{BPP}$, we allow the computational paths of a randomised algorithm to have different lengths, so that an input is accepted if there are at least twice as many possible accepting paths than rejecting paths, and is rejected if there are at least twice as many possible rejecting paths than accepting paths. In this view we no longer use post-selection, and instead design our algorithms so that the number of paths that return the correct answer sufficiently overwhelm the number of paths that return the incorrect answer. For this reason $\PostBPP$ is sometimes called $\mathsf{BPP}_{\mathsf{path}}$. 

It is known that $\PostBPP$ contains $\mathsf{MA}$ and $\mathsf{P}^{\mathsf{NP[log]}}$ ($\mathsf{P}$ with an $\mathsf{NP}$ oracle that can be called a logarithmic number of times), and is contained in $\PP$, $\mathsf{BPP}^{\mathsf{NP}}$, and the third level of the polynomial hierarchy \cite{han1997threshold}. Hence, $\PostBPP$ lies somewhere between the first and third levels of the polynomial hierarchy. Its quantum counterpart, $\PostBQP$, was shown by Aaronson~\cite{aaronson2005quantum} to be equal to the class $\PP$. Toda \cite{toda1991pp} showed that any problem in the polynomial hierarchy ($\mathsf{PH}$) can be reduced to solving a polynomial number of problems in $\PP$, and hence that $\mathsf{PH} \subseteq \mathsf{P}^{\PP} $. Therefore, quantum post-selection appears to be more powerful than classical post-selection, but the precise difference in computational power is not entirely clear. Indeed, if $\PostBPP = \PostBQP$, then the polynomial hierarchy collapses -- something that is believed to be extremely unlikely. Conversely, the non-collapse of the polynomial hierarchy would imply that the computational powers of classical post-selection and quantum post-selection are far apart. 

This observation forms the basis of many results in the area of quantum computational supremacy\footnote{Quantum computational supremacy refers to the experimental realisation of a (perhaps non-universal) quantum computer that can unambiguously outperform any existing (or soon to exist) classical computer for some task, and is considered to be an important milestone in the field of quantum computation.}~\cite{harrow2017quantum,lund2017quantum}. In particular, if it were possible to efficiently simulate a quantum computer using a classical computer, then this would imply that a classical computer with post-selection could efficiently simulate a quantum computer with post-selection (i.e. that $\PostBPP = \PostBQP$) which, as noted above, is generally believed to be unlikely to be true. Moreover, many non-universal models of quantum computation become universal once post-selection is allowed~\cite{bremner2010classical}, and so even an efficient classical simulation of one of these restricted models of quantum computing would lead to a collapse of the polynomial hierarchy \cite{harrow2017quantum}. Therefore, the study of both classical and quantum computation with post-selection might yield interesting insights into the differences between classical and quantum computation in general.

~\\
In this work, we study the query analogue of the class $\PostBPP$. The query analogues of complexity classes can often shed light on the differences between various models of computation: for instance, it is possible to rigorously prove separations between the power of quantum and classical computation in the query model~\cite{aaronson2016separations}; however, we do not know how to prove such strong results in the context of circuit complexity, where important questions remain open. 

In the query model, we count only the number of times that an algorithm needs to access the bits of the input string to be able to compute some (usually Boolean) function, and not the total computation time. Recently, the query complexity of post-selected quantum computation was studied by Mahadev and de Wolf \cite{mahadev2015rational}: they showed that the post-selected quantum query complexity ($\mathsf{PostQ}$) of a Boolean function is equal, up to a constant factor of 2, to the minimum degree of a rational function that approximates that function. The link between post-selection and rational approximation was first pointed out by Aaronson \cite{aaronson2005quantum}, but was made rigorous in \cite{mahadev2015rational}. Here, we study the classical analogue of this query model, and find that the classical post-selected query complexity (which we name $\mathsf{PostR}$) is related in a similar way to the degrees of rational functions that have only positive coefficients. 

~\\
We remark that our lower bounds carry over to the communication complexity setting (see Section \ref{sec:comms} for a definition of this model). Post-selected classical communication complexity, which we write $\PR^{\mathsf{CC}}$, has been studied previously: as ``approximate majority covers'' in \cite{klauck2003rectangle}, and ``zero-communication protocols'' in \cite{gavinsky2014route}, where the term ``extended discrepancy'' was coined for the dual characterisation. G\"o\"os~et~al.~\cite{goos2016rectangles} proved a so-called `simulation theorem', showing that lower bounds in certain query models can yield lower bounds in their communication analogues. Our characterisation of post-selected classical query complexity in terms of the degrees of rational functions can therefore also be used to derive lower bounds in the communication complexity setting.

\subsection{Organisation}
Section \ref{sec:defs} provides some definitions that will be useful throughout the paper. In Section \ref{sec:overview}, we describe our main results, and then in Section \ref{sec:proofs} we prove these results. In Section \ref{sec:lower_bound}, we prove lower bounds on the post-selected classical query complexity of approximating the Majority function, and describe an optimal (up to constant factors) algorithm for doing so. In Appendix \ref{app:generalisation} we generalise the lower bound to arbitrary symmetric functions. Finally, in Section \ref{sec:counting} we describe an efficient post-selected query algorithm for approximate counting.

\section{Definitions}\label{sec:defs}
We begin by introducing and defining some concepts that will be helpful in understanding our results. Here and elsewhere, we use the notation $[N]$ to represent the set of integers $\{1,\dots,N\}$, and we write $|x|$ to denote the Hamming weight of the bit-string $x$. For an overview of complexity measures for Boolean functions and their relationships to each other, see for example the survey paper by Buhrman and de Wolf~\cite{buhrman1999bounds}.

\subsection{Polynomial approximations}
In this work we will concentrate on $N$-variate polynomials in which the input variables take on values from $\{0,1\}$. Such an $N$-variate polynomial is a function $P : \{0,1\}^N \rightarrow \mathbb{R}$, which can be written as $P(x) = \sum_{S \subseteq [N]} \alpha_S X_S$, where $X_S = \prod_{i \in S} x_i$, and the $\alpha_i$ are real coefficients. The \emph{degree} of $P$ is $\deg(P) = \max \{ |S| : \alpha_S \neq 0 \}$. Since $x_i^k = x_i$ for $k \geq 1$, we can restrict our attention to \emph{multilinear polynomials}, where the degree of each variable is at most 1 (and therefore the degree of any polynomial is at most $N$). 

For $\epsilon \in [0,1/2)$, we say that a polynomial $P$ $\epsilon$-approximates a function $f : D\subseteq \{0,1\}^N \rightarrow \mathbb{R}$, if $|f(x) - P(x)| \leq \epsilon$ for all inputs $x \in D$. The $\epsilon$-approximate degree of $f$, abbreviated $\deg_\epsilon(f)$, is the minimum degree over all polynomials that $\epsilon$-approximate $f$. The exact degree of $f$ is usually written as $\deg(f)$, and the `approximate degree' of $f$, for some constant $\epsilon$ bounded away from $1/2$ (usually $\epsilon$ is chosen to be $1/3$) is written as $\widetilde{\deg}(f)$.

\subsubsection{Non-deterministic polynomials}
We will also have reason to consider \emph{non-deterministic polynomials}. A non-deterministic polynomial for $f$ is a polynomial $P : \{0,1\}^N \rightarrow \mathbb{R}$ which is non-zero iff $f(x)=1$. The non-deterministic degree of a function $f$ is the smallest degree of a non-deterministic polynomial for $f$, and is written $\ndeg(f)$.

\subsubsection{Nonnegative literal polynomials}\label{sec:ldeg}
As well as non-deterministic polynomials, we will briefly consider \emph{nonnegative literal polynomials}, introduced by Kaniewski et al. in~\cite{kaniewski2015query}. These are $2N$-variate polynomials over the variables $\{x_i, (1-x_i) : i \in [N]\}$ with only nonnegative coefficients. For a function $f : \{0,1\}^N \rightarrow \mathbb{R}^+$, let the nonnegative literal degree for `yes'-instances, $\mathsf{ldeg}^+_1(f)$, be the minimum degree over nonnegative literal polynomials that \emph{equal} $f(x)$ for all $x\in\{0,1\}^N$. Similarly, let the nonnegative literal degree for `no'-instances, $\mathsf{ldeg}^+_0(f)$ be the minimum degree over nonnegative literal polynomials that equal $1-f(x)$ for all $x\in\{0,1\}^N$. Finally, define $\mathsf{ldeg}^+(f) = \max\{\mathsf{ldeg}_0^+(f)$, $\mathsf{ldeg}_1^+(f)\}$. The reason for this slightly non-standard definition of polynomial degree should become apparent later in this section.

\subsection{Rational functions}
A \emph{rational function}, or \emph{rational polynomial}, $R = P/Q$ is the ratio of two $N$-variate polynomials $P$ and $Q$, where $Q$ is defined to be non-zero everywhere to avoid division by zero. The degree of $R$ is defined to be the maximum of the degrees $P$ and $Q$. As before, we say that a rational function $\epsilon$-approximates a function $f$ if $|P(x)/Q(x) - f(x)| \leq \epsilon$ for all $x \in D$. The $\epsilon$-approximate rational degree $\rdeg(f)$ of $f$ is the minimum degree over all rational functions that $\epsilon$-approximate $f$. In this work, we will consider $2N$-variate rational functions over the variables $\{x_i, (1-x_i) : i \in [N]\}$ in which all of the coefficients of $P$ and $Q$ are nonnegative, and call such a function a rational function with positive coefficients (note that these are not as general as rational functions with arbitrary coefficients -- indeed, if $f$ is a rational function with positive coefficients then $f(x) \geq 0$ for $x \in [0,1]$). We denote by $\rdeg_\epsilon^+(f)$ the minimum degree over all rational functions with positive coefficients that $\epsilon$-approximate $f$.

\subsection{Post-selected classical query algorithms}

Here we define the complexity measure $\PR_\epsilon$: the query complexity of post-selected (randomised) classical computation with error $\epsilon$. First we precisely define what we mean by a classical post-selected query algorithm.
\begin{definition}
A classical post-selected query algorithm $\mathcal{A}$ consists of a probability distribution over deterministic decision trees\footnote{A (deterministic) decision tree is an adaptive algorithm for computing Boolean functions that chooses which input variable to query next based on the answers to the previous queries.}
that can output $0, 1,$ or $\perp$ (`don't know'). Given an input $x$, $\mathcal{A}$ chooses a decision tree from its distribution (independently of $x$), and outputs the answer of that decision tree on input $x$. For a Boolean function $f$, we require that the output of $\mathcal{A}$ on $x$ satisfies
\[
\Pr[\mathcal{A}(x) = f(x) | \mathcal{A}(x) \neq \perp] \geq 1-\epsilon
\]
for some choice of $\epsilon \in [0,1/2)$.
\end{definition}
Note that this definition is similar to the one given for $\PostBPP$ in Section \ref{sec:intro}, in which we consider a separate algorithm $\mathcal{B}$ to be the `post-selector'. It is easy to see that the two views are equivalent -- we can consider a single algorithm/decision tree that outputs $0, 1,$ or $\perp$; or two algorithms/decision trees that each output $0$ or $1$. 
\\
\\
Now we can define the complexity measure $\PRe$.
\begin{definition}\vspace{-0.5cm}
The $\epsilon$-approximate post-selected classical query complexity of a Boolean function $f$, written $\PRe(f)$, is the minimum query complexity over all classical post-selected query algorithms that $\epsilon$-approximate $f$. 
\end{definition}
One might ask what happens if we are allowed multiple rounds of post-selection: that is, what if we allow a post-selected query algorithm to be used as a subroutine inside another post-selected query algorithm? In this case, every Boolean function can be computed up to constant bounded error using only 1 query to the input: we can construct a post-selected query algorithm that decides, up to arbitrarily small one-sided error, whether some `guessed' string $y$ equals the input $x$. Then, using this as a subroutine, we can guess a string $y$ and post-select on it being equal to $x$. Then we output $f(y)$, which will be correct with bounded error (of, say, $1/3$), and requires only a single query to $x$ (made by the single call to the subroutine). A more rigorous proof of this observation is given in Appendix \ref{app:nested}.

\subsection{Certificate complexity and non-deterministic query complexity}\label{sec:certificate_def}
Let $f: D \subseteq \{0,1\}^N \rightarrow \{0,1\}$ be an $N$-bit Boolean function. Then we have the following definitions:
\begin{definition}[Certificate complexity]
A $b$-certificate for an input $x$ is an assignment $C_b : S \rightarrow \{0,1\}$ to some set $S \subseteq [N]$ of variables, such that $f(x) = b$ whenever an input $x$ is consistent with $C_b$. The size of the certificate $C_b$ is $|S|$. The $b$-certificate complexity $\C_b^x(f)$ of $f$ on input $x$ is the minimal size of a $b$-certificate that is consistent with $x$ (where $b=f(x)$). The $1$-certificate complexity is $\C_1(f) = \max_{x:f(x)=1} C^x(f)$. The $0$-certificate complexity $\C_0(f)$ is defined analogously. Finally, we define the certificate complexity of $f$ to be $\C(f) = \max\{\C_0(f), \C_1(f)\}$.
\end{definition}

\begin{definition}[Non-determinstic query algorithms]
A non-deterministic query algorithm for $b$-instances is a randomised algorithm whose acceptance probability is positive if $f(x) = b$, and zero if $f(x)=1-b$. Then the non-deterministic query complexity on input $x$, $\N_{f(x)}(f)$, is the minimum number of queries required by a classical algorithm to achieve the above behaviour on input $x$. The non-deterministic query complexity for $1$-instances is $\N_1(f) = \max_{x:f(x)=1} \N_1(f)$, and the non-deterministic query complexity for $0$-instances, $\N_0(f)$, is defined analogously. Finally, the non-deterministic query complexity of $f$ is defined to be $\N(f) = \max\{\N_0(f),\N_1(f)\}$.

If we replace the classical query algorithm with a quantum query algorithm, we obtain the quantum counterparts of the above complexity measures $\NQ_1, \NQ_0$ and $\NQ$.
\end{definition}
~\\
The following results from \cite{de2000characterization} and \cite{de2003nondeterministic} will prove useful later on:
\begin{lemma}[\cite{de2000characterization}, Proposition 1]\label{cor:NC}
For any Boolean function $f$, 
\[
\C_b(f) = \N_b(f)
\]
for $b\in\{0,1\}$. Therefore,
\[
\C(f) = \N(f).
\]
\end{lemma}
\begin{lemma}[\cite{de2003nondeterministic}, Theorem 2.3]
For any Boolean function $f$, 
\[
\NQ(f) = \ndeg(f) \leq N(f).
\]
\end{lemma}
\noindent Finally, we have the following useful result:
\begin{lemma}[\cite{de2000characterization}]\label{lem:dewolf2}
Let $f$ be a non-constant symmetric function on $N$ variables. Suppose $f$ achieves value $0$ on $z$ Hamming weights $k_1, \dots, k_z$. Then $\ndeg(f) \leq z$.
\end{lemma}
\begin{proof}
Since $|x| = \sum_j x_j$, then $(|x| - k_1)(|x| - k_2)\cdots(|x| - k_z)$ is a non-deterministic polynomial for $f$ with degree at most $z$. 
\end{proof}
~\\
In \cite{kaniewski2015query}, Kaniewski et al. show that the minimal degree of (literal) polynomials with positive coefficients (i.e. those described in Section \ref{sec:ldeg}) exactly characterises the model of `classical query complexity in expectation': in this model, we consider (classical) query algorithms that, on input $x$, output a nonnegative random variable whose \emph{expectation} equals $f(x)$. For a function $f : \{0,1\}^N \rightarrow \mathbb{R}^+$, $\mathsf{RE}(f)$ is the smallest number of queries that an algorithm needs to make to obtain the above behaviour for all $x \in \{0,1\}^N$. 
In particular, and using the notation introduced in this section, they show:
\begin{lemma}[\cite{kaniewski2015query}, Theorem 9]
\label{lem:ldeg}
Let $f:\{0,1\}^N \rightarrow \mathbb{R}^+$. Then $\mathsf{RE}(f) = \mathsf{ldeg}^+_1(f)$.
\end{lemma}
\noindent In the special case that $f$ is a Boolean function, a query algorithm that computes $f$ in expectation is (a slightly more constrained version of) a non-deterministic algorithm for yes-instances of $f$. Combining Lemma \ref{cor:NC} with a straightforward generalisation of this result, we have, for any Boolean function $f$,
\begin{corollary}\label{cor:ldeg}
$\C(f) = \N(f) \leq \mathsf{ldeg}^+(f)$.
\end{corollary}

\section{Overview of Results}\label{sec:overview}
Our main result is a tight connection between post-selected classical query algorithms and rational functions with positive coefficients. This result is essentially the classical analogue of the one shown by Mahadev and de Wolf in \cite{mahadev2015rational}. In particular, we show
\begin{restatable}{theorem}{maintheorem}
\label{theo:main}
For any Boolean function $f : D \subseteq \{0,1\}^N \rightarrow \{0,1\}$, and $\epsilon \in (0,1/2)$ we have
\[
\rdeg_\epsilon^+(f) \leq \PRe(f) \leq 2 \rdeg_\epsilon^+(f).
\]
\end{restatable}
That is, the post-selected classical query complexity of a Boolean function is essentially equivalent to the minimal degree of a rational function with positive coefficients that approximates that function. As in the quantum case (\cite{mahadev2015rational}), we obtain a tight relation between the two complexity measures. This can be contrasted to the ordinary case (i.e. without post-selection), where the approximate degree $\widetilde{\deg}(f)$ of a function is only known to be polynomially related to the randomised query complexity $\mathsf{R}(f)$, and where polynomial gaps have indeed been shown to exist \cite{aaronson2016separations}.
\\
\\
Our next result concerns the special case of exact approximation:
\begin{restatable}{theorem}{zeroerror}
\label{theo:zero_error}
For any Boolean function $f : D \subseteq \{0,1\}^N \rightarrow \{0,1\}$ we have 
\[
\PR_0(f) = \N(f) = \C(f).
\]
\end{restatable}
\noindent Combining this result with Corollary \ref{cor:ldeg}, we have an upper bound on the complexity of zero-error post-selected classical query algorithms:
\begin{corollary}
$\PR_0(f) \leq \mathsf{ldeg}_+(f)$.
\end{corollary}
~\\
\noindent
There is a long-standing open question, attributed to Fortnow, but given by Nisan and Szegedy in \cite{nisan1994degree}, which asks: is there a polynomial relation between the \emph{exact} rational degree of a (total) Boolean function $f$ and its usual polynomial degree? In \cite{mahadev2015rational}, Mahadev and de Wolf recast this question in the context of post-selection algorithms: can we efficiently simulate an exact quantum algorithm with post-selection by a bounded-error quantum algorithm without post-selection? 

By considering exact \emph{classical} algorithms with post-selection, it is tempting to say that we have resolved this question in the case where the rational approximation can have only positive coefficients, but, unfortunately, the results of Theorem \ref{theo:main} break down when $\epsilon = 0$. What we \emph{can} state is the following:
\begin{corollary}
Given a degree-$d$ rational function with only positive coefficients that exactly represents a Boolean function $f$, it is possible to construct a post-selected classical query algorithm $\mathcal{A}$ with query complexity at most $d$, such that when $f(x) = 1$, $\mathcal{A}(x) = 1$ with certainty, and when $f(x)=0$, $\Pr[\mathcal{A}(x) = 0] \geq \frac{1}{1+\delta}$ for arbitrary $\delta>0$. 
\end{corollary}
\noindent That is, we can obtain an algorithm with arbitrarily small one-sided error, but not an exact algorithm. 

\subsection{Separations}\label{sec:seps}
It is often instructive to find explicit functions that exhibit a gap between complexity measures. Here we seek functions that separate post-selected classical query complexity from various other complexity measures. In our case, the $OR$ function\footnote{Defined as $OR(x) = \begin{cases} 0 & \text{if}~ x = 00\dots0 \\ 1 & \text{o.w.} \end{cases}$ } provides a number of useful separations. A post-selection algorithm for the $N$-bit $OR$ function is as follows:
\begin{itemize}
\item Choose a bit $x_i$, $i \in [N]$ uniformly at random from the input.
\item If $x_i = 1$, return 1.
\item Else return $0$ with probability $1/(2N)$.
\item Else return $\perp$ (`don't know').
\end{itemize}
We post-select on not seeing $\perp$ as the outcome. In the case that the input is the all-zero string, then the algorithm can only ever return $0$, conditioned on it not returning $\perp$. In the case that at least one bit is equal to $1$, the probability that the algorithm returns 1 is
\[
\Pr[\text{return }1|\text{not return}\perp] = \frac{\Pr[\text{return }1]}{\Pr[\text{not return}\perp]} = \frac{|x|/N}{|x|/N + (1-|x|/N)\cdot 1/(2N)} \geq 2/3.
\]
Therefore, we can compute the $OR$ function up to (one-sided) error $1/3$ using a single query to the input, and hence $\PRR(OR) = 1$. By combining this with existing results, we obtain the following separations:
\begin{itemize}
\item A super-exponential separation from bounded-error quantum query complexity ($\Q$), since $\Q(OR) = \Theta(\sqrt{N})$ due to Grover's algorithm \cite{grover1996fast}.
\item An unbounded separation from \emph{exact} post-selected classical query complexity: since $\C_1(OR) = 1$ and $\C_0(OR) = N$, then $\C(OR) = N$ and so by Theorem \ref{theo:zero_error}, $\PR_0(OR)~=~N$.
\item A super-exponential separation from quantum certificate complexity $\mathsf{QC}$\footnote{See \cite{aaronson2003quantum} for a definition of quantum certificate complexity.}, since $\mathsf{QC}(OR) = \Omega(\sqrt{N})$, as shown by Aaronson \cite{aaronson2003quantum}. This separation is interesting, since $\mathsf{QC}$ is the query analogue of the class $\mathsf{QMA}$, and it is not clear whether $\PostBPP$ is more powerful than $\mathsf{QMA}$.
\end{itemize}
Hence, post-selection makes the $OR$ function trivial to compute up to bounded error, and shows that post-selected algorithms (both quantum and classical) can be much more powerful than classical, quantum, and non-deterministic (or exact post-selected) query algorithms. 
The $OR$ function is an example of a problem that is easy for both quantum and classical post-selection. Indeed, a simple degree-$1$ rational function for estimating the $OR$ function up to bounded error is 
\[
P_{OR}(x) = \frac{\sum_{i=1}^N x_i}{\epsilon + \sum_{i=1}^N x_i},
\]
for some small and fixed constant $\epsilon > 0$. This polynomial has only positive coefficients, and thus allows for an approximation by a classical post-selected query algorithm. In order to separate post-selected classical query complexity and post-selected quantum query complexity, we must find a function that allows for a low-degree rational approximation, but does not allow for a low-degree approximation by rational functions with only positive coefficients. 

To this end, in Section \ref{sec:lower_bound} we prove a $\Omega(N)$ lower bound on the post-selected classical query complexity of approximating the Majority function, defined on $N$-bit strings as:
\[
MAJ_N(x) = 
\begin{cases}
1 \qquad \text{if $|x| > N/2$} \\
0 \qquad \text{if $|x| \leq N/2$} \\
\end{cases}.
\]
This gives an exponential separation between quantum and classical post-selected query complexities. In particular, we show
\begin{theorem}
$\PRR(MAJ_N) = \Theta(N)$.
\end{theorem}
\noindent Since $\PQ(MAJ) = \Theta(\log N)$~\cite{mahadev2015rational} and $\PR(MAJ) = \Theta(N)$, this result shows that post-selected quantum computation can be much more powerful than post-selected classical computation in the query model. 
\\
\\
Following this, we generalise the lower bound to any symmetric function $f$. Write $f_k = f(x)$ for $|x| = k$, and define
\[
\Gamma(f) = \min\{|2k - N + 1| : f_k \neq f_{k+1}, 0 \leq k \leq N-1,
\]
Then we show:
\begin{theorem}\label{theo:generalised1}
\[
\PR(f) \geq \frac{1}{8}\left(N - \Gamma(f)\right),
\]
for all non-constant symmetric Boolean functions $f$.
\end{theorem}
\noindent
This result is similar to a result of Paturi~\cite{paturi1992degree}:
\begin{theorem}[Paturi]
If $f$ is a non-constant symmetric Boolean function on $\{0, 1\}^N$, then $\widetilde{\deg}(f) =  \Theta(\sqrt{N(N-\Gamma(f))})$.
\end{theorem}
\noindent
The lower bound from Theorem \ref{theo:generalised1} above can be written in a similar form:
\[
\rdeg^+(f) = \Omega(N - \Gamma(f))
\]
for any (non-constant) symmetric function $f$. 
\\
\\
We can use this characterisation to separate non-deterministic quantum query complexity and post-selected classical query complexity. Consider the $N$-bit function
\[
f(x) = 
\begin{cases}
0 & \text{if } |x| = \lceil N/2\rceil \\
1 & \text{o.w.} 
\end{cases}
\]
By Lemma \ref{lem:dewolf2}, we know that $\NQ(f) \leq \ndeg(f) \leq 1$, since $f$ is 0 on only one Hamming weight $|x| = \lceil N/2\rceil$. On the other hand, we have $\Gamma(f) = O(1)$ and so by Theorem \ref{theo:generalised1}, $\PR(f) = \Omega(N)$. This gives an unbounded separation between non-deterministic quantum query algorithms and post-selected classical query algorithms, and also implies that $\ndeg(f) \leq \rdeg^+(f)$.

\subsection{Relation to communication complexity}\label{sec:comms}
Our results have some interesting connections to the communication complexity model. In this model, we consider two parties, Alice and Bob, who together want to compute some function $f : D \rightarrow \{0, 1\}$, where $D \subseteq X \times Y$. Alice receives input $x \in X$, Bob receives input $y \in Y$, with $(x, y) \in D$. Typically, we choose $X = Y = \{0, 1\}^n$. As the value $f(x, y)$ will generally depend on both $x$ and $y$, some amount of communication between Alice and Bob is required in order for them to be able to compute $f(x, y)$. If $D = X \times Y$, then the function $f$ is called \emph{total}, otherwise it is called \emph{partial}. The communication complexity of $f$ is then the minimal number of bits that must be communicated in order for Alice and Bob to compute $f$.

In the introduction, we mentioned a method for obtaining lower bounds in communication complexity from lower bounds in query complexity, via the `simulation theorem' of G\"o\"os et al.~\cite{goos2016rectangles}. More precisely, they show that any $\PR^{\mathsf{CC}}$ protocol for the composed function $f \circ g^N$ (where $g$ is a particular small two-party function, often called a `gadget'\footnote{In this case, $g$ is chosen to be
\begin{itemize}
\item $g(x, y) := \langle x, y\rangle \mod 2$, where $x, y \in \{0, 1\}^b$
\item The block length $b = b(N)$ satisfies $b(N) \geq 100 \log N$.
\end{itemize}}) can be be converted into a corresponding query algorithm for $f$, which implies that $\PR^{\mathsf{CC}} (f \circ g^N) \geq \Omega(\PR(f) \cdot \log N)$. The authors prove analogous results for other models of computation, in which the lower bound becomes an equality (up to constant factors). Interestingly, for post-selection only a lower bound was shown, and so we cannot use our results to completely characterise the communication analogue of $\PR$. 

In Section \ref{sec:lower_bound}, we prove a $\Omega(N)$ lower bound on the post-selected classical query complexity of approximating the Majority function. In \cite{klauck2003rectangle}, Klauck gives an analogous $\Omega(N)$ lower bound on the post-selected classical \emph{communication complexity} of approximating the Majority function. We note that this lower bound could also be obtained by combining the lower bound presented in this work with the fact that lower bounds on query complexity imply lower bounds on communication complexity in the presence of post-selection. 

The reader familiar with communication complexity might wonder why Klauck's lower bound doesn't immediately imply the query lower bound, since Alice and Bob can run a communication protocol in which they just simulate the query algorithm, and hence the communication complexity is upper bounded by the query complexity. However, in order to simulate (post-selected) query algorithms in this way we require that the two parties have access to shared randomness, whereas Klauck's lower bound assumes that the two parties only have access to private randomness. This is a somewhat necessary assumption -- if we allow for shared randomness, then all Boolean functions have $O(1)$ communication complexity \cite{goos2016rectangles} unless we incorporate an extra charge of $\log(1/\alpha)$ into the communication cost, where $\alpha$ is the probability that the query algorithm returns `don't know'. 

On the other hand, Klauck~\cite{klauck2003rectangle} gives an $O(\log N)$ \emph{upper bound} for approximating the Majority function in the communication complexity analogue of the class $\PP$. This upper bound is analogous to the one from \cite{mahadev2015rational} for the query complexity analogue of the class $\PostBQP=\PP$, and hence we suspect that Sherstov's~\cite{sherstov2009intersection} lower bound on the rational polynomial degree of the Majority function also carries over into the communication complexity setting.

\subsection{Techniques}
The proof of our main result (Theorem \ref{theo:main}) -- namely that any degree-$d$ rational function with positive coefficients can be approximated by a $O(d)$-query post-selected classical query algorithm -- uses the observation that we can treat such rational functions as probability distributions over sets of monomials. Hence, given a rational function, we can sample from the set of monomials and use the probability amplification powers of post-selection to accurately approximate its value. On the other hand, if the coefficients of the rational functions can be negative, or take complex values, then the view that the coefficients represent probabilities can be replaced by the view that they are `amplitudes' that can interact in complicated ways to produce more complex behaviour. This view is used in \cite{mahadev2015rational}, in which post-selected \emph{quantum} algorithms are used to approximate rational functions by constructing quantum states whose amplitudes are proportional to the coefficients of the monomials of the rational functions.  Hence, the only difference between classical and quantum post-selected query algorithms is in the use of complex amplitudes over conventional probabilities -- an intuition that fits nicely with our understanding of quantum mechanics. 
\\
\\
In order to prove the lower bounds in Section \ref{sec:lower_bound}, we essentially bound how fast a low-degree (univariate) rational function can grow around a certain threshold. Similarly to other lower bounds on the degrees of polynomials that approximate threshold functions (e.g. \cite{paturi1992degree}), we find that the nearer this threshold is to the `middle', the faster the rational function needs to grow, and hence the larger its degree needs to be. Crucially, without the presence of negative coefficients, there can be no `cancelling out' between the monomials of the polynomials that form the rational function, and therefore their degrees must be large enough to allow for rapid growth from 0 to 1 around the threshold value. 
\\
\\
Many of the other results in the paper involve the construction of post-selected query algorithms. Each algorithm follows the same general structure:
\begin{itemize}
\item Query some input bits.
\item Based on these bits (for example, we might run some deterministic algorithm with these bits as input), either:
	\begin{itemize}
	\item Accept.
	\item Reject with some fixed small probability.
	\end{itemize}
\item Otherwise return `don't know'.
\item Finally, post-select on not seeing `don't know'.
\end{itemize}
The power of post-selection lies in the ability to decide to return `don't know' rather than reject, which allows us to drastically amplify the success probability of the underlying algorithm. For instance, suppose that we want to compute the $OR$ function on $N$ bits, and are therefore trying to determine if the input string $x$ is all zeroes, or if it contains at least one 1. In the first step, we can query a single element. If this element is 1, then we know for certain that $OR(x) = 1$, and so we return 1. In the case that this bit is not 1, then we still don't know whether $OR(x) =$ 0 or 1; however, we have gained some information about the input: at least one bit is 0. With complete ignorance about the rest of the input string, we can assign a probability to the event that the other $N-1$ bits of the input string are also zero: namely, $\frac{1}{2^{N-1}}$. So, with probability $\frac{1}{2^{N-1}}$ we return 0, and otherwise we return `don't know'. In this way, if there is at least one input bit set to 1, then we can amplify the probability that the algorithm will return 1 to a constant greater than, say, $2/3$.
\\
\\
An interesting phenomenon of success probability amplification in post-selected query algorithms was observed by Sherstov~\cite{sherstov2009intersection}, in the context of rational approximations to Boolean functions. 
\begin{lemma}\label{lem:sherstov}[Sherstov~\cite{sherstov2009intersection}]
Denote by $R(f,d)$ the smallest error achievable by a degree-$d$ rational approximation to the Boolean function $f: X \subseteq \mathbb{R}^N \rightarrow \{0,1\}$. Then for all integers $d>1$ and reals $t > 2$, 
\[
R(f, td) \leq R(f,d)^{t/2}.
\]
\end{lemma}
Informally, this result says the following: given a rational function that $\epsilon$-approximates a function $f$, we can obtain a new rational function that $\epsilon^2$-approximates $f$ by increasing the degree by at most a constant factor of 4. We can then repeatedly apply this procedure to amplify the probability of success further. Sherstov describes how to construct the new rational function without introducing any negative coefficients, so by Theorem \ref{theo:main} we would expect to see a similar amplification property for post-selected classical (and quantum) query algorithms. Indeed, we implicitly make use of this observation in our proof of Theorem \ref{theo:main}.

\section{Proofs of the main results}\label{sec:proofs}

\subsection{Exact post-selected query complexity $=$ certificate complexity}
Here we show that when post-selected query algorithms are not allowed to make mistakes, they are equivalent to non-deterministic query algorithms.
\zeroerror*
~\\
This result follows from the following upper and lower bounds, and the fact that $\N=\C$~\cite{de2000characterization}.
\begin{itemize}
\item $\PR_0 \leq \C$,
\item $\PR_0 \geq \max\{\N_0, \N_1\} = \N$.
\end{itemize}

\subsubsection{$\PR_0 \leq \C$}
To prove this bound, we construct a $\PR_0$ algorithm $\A$ that can compute any Boolean function $f$ using at most $\C$ queries to the input. Given an $n$-bit input $x \in \{0,1\}^N$, we compute $f(x) \in \{0,1\}$ as follows:
\begin{itemize}
\item Choose a random certificate (which can be either a 1-certificate or a 0-certificate) $C$.
\item If $x$ is consistent with $C$, then output 0 if $C$ is a 0-certificate, or 1 if $C$ is a 1-certificate.
\item Otherwise, output $\perp$. 
\end{itemize}
Then we post-select on not seeing $\perp$. The query complexity follows from the fact that the algorithm only checks one certificate, whose size will be at most $\C$. To prove correctness, consider the case where $f(x)=0$. In this case, at least one of the 0-certificates will be consistent with $x$, and none of the 1-certificates will be. Let the number of consistent 0-certificates be $c$. Then 
\[
\Pr[\A(x) = 0 | \A(x)\neq\perp] = \frac{\Pr[\A(x) = 0]}{\Pr[\A(x) \neq \perp]} = \frac{c/2^N}{c/2^N} = 1.
\]
Where the value of the denominator follows from the fact that only a 0-certificate can be consistent with $x$ and cause the algorithm to not return $\perp$. A similar argument holds in the case that $f(x)=1$. Hence, the algorithm will return the value of $f(x)$ with certainty in all cases. 

\subsubsection{$\PR_0 \geq \max\{\N_0, \N_1\} = \N$}
Here we show that, given a zero-error classical post-selected query algorithm, it is possible to construct a non-deterministic classical query algorithm whose query complexity is equal to the query complexity of the post-selection algorithm.

Suppose that we are given the description of a $\PR_0$ algorithm for computing some boolean function $f$. Such an algorithm is defined by a probability distribution over a set of decision trees that each output $0, 1$, or $\perp$ for some input $x$. If we replace the $\perp$ output with $0$, then we obtain a non-deterministic algorithm for `yes' instances: any tree chosen from the distribution will either output $f(x)$ or $0$, thus the probability of the algorithm accepting is non-zero if and only if $f(x)=1$. Likewise, if we replace the $\perp$ output with $1$, then we obtain a non-deterministic algorithm for `no' instances. In both cases, the decision trees themselves have not been changed, and therefore the query complexity of the non-deterministic algorithm is the maximum of the query complexity of any of the decision trees, which is just the query complexity of the post-selection algorithm. This is enough to show that $\PR_0(f) \leq \max\{\mathsf{N_0}, \mathsf{N_1}\} = \mathsf{N}$ for all $f$. 

By Lemma \ref{cor:NC}, we have $\mathsf{N} = \mathsf{C}$ (i.e. the non-deterministic query complexity is the same as the certificate complexity) \cite{de2000characterization}, and so $\PR_0 = \mathsf{N} = \mathsf{C}$. This result is perhaps a little surprising -- classical post-selection provides no advantage over non-deterministic query algorithms when it is not allowed to make mistakes.

\subsection{Bounded-error post-selected query algorithms are more powerful than non-deterministic query algorithms}
\subsection{$\PRR \leq \min\{\C_0, \C_1\}$}
Here we show that classical post-selection can have a (much) smaller query complexity than non-deterministic query algorithms when it is allowed to make mistakes. In particular, we consider the case where the post-selection algorithm must return $f(x)$ with probability $\geq 2/3$. 

Given some input $x \in \{0,1\}^N$, we show that it is possible to compute $f(x) \in \{0,1\}$ up to bounded-error using at most $\min\{\C_0, \C_1\}$ queries with classical post-selection. First we show that $f(x)$ can be computed using at most $\C_1$ queries, using the following post-selection algorithm:
\begin{itemize}
\item Choose a random 1-Certificate $C$.
\item If $x$ is consistent with $C$, then output 1.
\item Else output $0$ with probability $1/2^{N+1}$.
\item Otherwise, output $\perp$.
\end{itemize}
Then we post-select on not obtaining $\perp$. The algorithm only makes use of a single $1$-Certificate, and hence makes at most $\C_1$ queries to the input.

To show correctness, first suppose that $f(x) = 1$. Then there must exist at least one 1-Certificate that is consistent with $x$. Since there are at most $2^N$ 1-Certificates for $f$, then the probability that at least one randomly chosen 1-Certificate is consistent with $x$ is $\geq1/2^N$. Then the probability of seeing $1$, post-selecting on not seeing $\perp$, is given by
\begin{eqnarray*}
\Pr[\text{see } 1 | \text{not see } \perp] &=& \frac{\Pr[\text{see } 1]}{\Pr[\text{not see } \perp]} \\
&\geq& \frac{\frac{1}{2^N}}{\left(1-\frac{1}{2^N}\right)\frac{1}{2^{N+1}} + \frac{1}{2^N}} \\
&=& \frac{1}{\frac12 - \frac{1}{2^{N+1}} + 1} \\
&>& \frac{1}{\frac12 + 1} \\
&=& \frac{2}{3}.
\end{eqnarray*}
On the other hand, suppose that $f(x) = 0$. Then there is no 1-Certificate $C$ that is consistent with $x$. So, regardless of our choice of $C$, the algorithm always either returns $0$ with probability $1/2^{N+1}$, or otherwise returns $\perp$. Hence, in this case, we have
\begin{eqnarray*}
\Pr[\text{see } 0 | \text{not see } \perp] &=& \frac{\Pr[\text{see } 0]}{\Pr[\text{not see } \perp]} \\
&=& \frac{1/2^{N+1}}{1/2^{N+1}} \\
&=& 1.
\end{eqnarray*}
So, when $f(x) = 1$, the algorithm returns 1 with probability at least $2/3$. If $f(x)=0$, the algorithm returns $0$ with certainty. It does this by making at most $C_1$ queries to $x$, where $C_1$ is the maximum size of a 1-certificate for $f$. 
This implies that $\PRR(f) \leq \C_1(f)$. 
\\
\\
We can use a similar algorithm to show that $\PRR \leq \C_0$. Consider the following algorithm:
\begin{itemize}
\item Choose a random 0-Certificate $C$.
\item If $x$ is consistent with $C$, then output 0.
\item Else output $1$ with probability $1/2^{N+1}$.
\item Otherwise, output $\perp$.
\end{itemize}
By a similar proof to the previous case, if $f(x) = 0$, we return 0 with probability $>2/3$. If $f(x)=1$, we return 1 with certainty. Hence, $\PRR(f) \leq \C_0(f)$. 

Combining these cases, we have that $\PRR \leq \min\{\C_0, \C_1\} \leq \C$, with equality between the two terms on the right only when $\C_0 = \C_1$.

\subsection{Query complexity with classical post-selection $\approx$ degree of rational approximation with positive coefficients}
Here we prove our main result -- that the complexity of post-selected classical query algorithms is tightly characterised by the minimal degree of rational functions with positive coefficients. 
\maintheorem*

To prove this, we begin by showing that for any $d$-query classical post-selected query algorithm on the variables $\{x_1,\dots,x_N\}$, there is a corresponding degree-$d$ rational function over the variables $\{x_1,\dots,x_N\} \cup \{(1-x_1),\dots,(1-x_N)\}$ that has only positive coefficients.
\begin{lemma}
For all Boolean functions $f$, $\rdeg^+_\epsilon(f) \leq \PR_\epsilon(f)$
\end{lemma}
\begin{proof}

Consider a classical post-selected query algorithm with complexity $\PRe$ that computes some boolean function $f : \{0,1\}^N \rightarrow \{0,1\}$ with bounded error $\epsilon$. We take the view that this algorithm consists of a probability distribution $\sigma$ over a number of deterministic decision trees. Each decision tree $i$ computes two functions $g_i(x)$ and $h_i(x)$, where $x = x_1,\dots,x_N$. The first, $g$, is the tree's (proposed) answer to the Boolean function, and the second, $h$, is the post-selection function. Each decision tree must compute \emph{both} of these functions using at most $\PRe$ queries to the input. Denote the decision tree complexities of $g_i$ and $h_i$ by $\D(g_i)$ and $\D(h_i)$, respectively. Associated to each boolean function $g_i$ and $h_i$ are (unique) multivariate polynomials $p_i$ and $q_i$~\cite{nisan1994degree}, such that $\deg(p_i), \deg(q_i) \leq \PRe(f)$. The acceptance probability of the algorithm is given by 
\[
\Pr_i[g_i(x) = 1 | h_i(x) = 1] = \frac{\Pr_i[g_i(x) = 1]}{\Pr_i[h_i(x)=1]} = \frac{\Pr_i[p_i(x)=1]}{\Pr_i[q_i(x)=1]}.
\]
The latter term is just a rational function $P/Q$, where $P(x) = \sum_i \sigma(i)p_i(x)$ and $Q(x) = \sum_i \sigma(i)q_i(x)$. Since $g_i$ and $h_i$ are total Boolean functions, the associated polynomials $p_i$ and $q_i$ can be written as polynomials in the variables $\{x_1,\dots,x_N\} \cup \{(1-x_1),\dots,(1-x_N)\}$, and only positive coefficients. To see this, consider some depth-1 decision tree, which queries the input bit $x_i$, and then outputs $b \in \{0,1\}$ if $x_i=1$, or $1-b$ otherwise. Such a decision tree can be (exactly) represented by the degree-1 polynomial $d(x) = b x_i + (1-b)(1-x_i)$. Now consider a depth-$T$ decision tree $t$, which begins by querying the input bit $x_t$. To obtain the polynomial that represents this tree, we can start at the root and write the corresponding polynomial as $d_t(x) = x_t d_l(x) + (1-x_t) d_r(x)$, where $d_l(x)$ and $d_r(x)$ are the polynomials corresponding to the left and right sub-trees of $t$, respectively. In this way, the polynomial can be defined recursively in $T$ steps, where at each stage the polynomial increases in degree by at most $1$. When we reach a leaf of the tree, then the depth-1 case is used to determine the coefficient of each monomial in the polynomial. It follows that the resulting polynomial can be written as a degree-$T$ polynomial in the variables $\{x_1,\dots,x_N\} \cup \{(1-x_1),\dots,(1-x_N)\}$, and only positive coefficients.

Hence, we have that $P$ and $Q$ are both polynomials with positive coefficients, and therefore $P/Q$ is a rational function with positive coefficients that $\epsilon$-approximates $f$ and has degree $\rdeg^+_\epsilon(f) = \max_i\{\deg(p_i), \deg(q_i)\} = \max_i\{\D(g_i), \D(h_i)\} \leq \PRe(f)$.
\end{proof}
To illustrate, consider the following $\PR$ algorithm for approximating the $AND$ function\footnote{Defined as $AND(x) = \begin{cases} 1 & \text{if}~x = 11\dots1 \\ 0 & \text{o.w.} \end{cases}$} up to bounded error $1/3$:
\begin{itemize}
\item Choose a variable $x_i$ uniformly at random from $x$.
\item If $x_i = 0$, return 0.
\item Else with probability $\frac{1}{2N}$, return $1$.
\item Else return $\perp$.
\item Post-select on not seeing $\perp$.
\end{itemize}
Clearly, this algorithm makes one query to the input, and so we would expect the corresponding rational function to have degree 1. In this algorithm, each decision tree is of the following form:
\begin{itemize}
\item Query variable $x_i$, where $i$ is fixed in advance.
\item If $x_i=0$, return 0.
\item Else, either return $1$ or $\perp$.
\end{itemize}
We require that for each $i$, there are $2N-1$ trees that output $\perp$ in the final step, and only one that outputs $1$. For each $i$, if we choose a tree uniformly at random, the probability that it returns $1$ is $\Pr[\text{see $1$}] = \frac{1}{2N}x_i$. The probability that it doesn't return $\perp$ is $\Pr[\text{not see $\perp$}] = \frac{1}{2N}x_i + (1-x_i)$. Hence, choosing $i$ also uniformly at random, the probability that the post-selection algorithm returns $1$ is:
\[
\Pr[\text{see $1$} | \text{not see $\perp$}] = \frac{\Pr[\text{see $1$}]}{\Pr[\text{not see $\perp$}]} = \frac{\frac{1}{N}\sum_i \frac{1}{2N}x_i}{\frac{1}{N} \sum_i \frac{1}{2N}x_i + (1-x_i)} = \frac{\frac{1}{2N} \sum_i x_i}{\frac{1}{2N}\sum_i x_i + \sum_i (1-x_i) },
\]
which is a degree-$1$ rational function in $\{x_1,\dots,x_N\} \cup \{(1-x_1),\dots,(1-x_N)\}$. To check that this polynomial does indeed approximate the $AND$ function, consider the hardest case to distinguish, when $|x| = N-1$. In this case, we have 
\[
\Pr[\text{see $1$} | \text{not see $\perp$}] = \frac{\frac{1}{2N} (N-1)}{\frac{1}{2N} (N-1) + 1} = \frac{N-1}{N-1 + 2N} \leq 1/3.
\]
On the other hand, when $|x| = N$, we have $\Pr[\text{see $1$} | \text{not see $\perp$}] = 1$, and so this polynomial does indeed approximate the $AND$ function up to one-sided error of $1/3$.
\\
\\
The next step in our proof is to show that a post-selected query algorithm can accurately determine the value of rational functions with positive coefficients.
\begin{lemma}\label{lem:7}
For all Boolean functions $f$, $\PRe(f) \leq 2\rdeg^+_\epsilon(f)$.
\end{lemma}
\begin{proof}
Consider an $n$-bit Boolean function $f$ and a rational function $P/Q$ with degree $d = \rdeg_\epsilon(f)$ and positive coefficients that $\epsilon$-approximates $f$. That is, $|P(x)/Q(x) - f(x)| \leq \epsilon$ for all $x \in \{0,1\}^n$. In particular, for each $x$
\begin{itemize}
\item If $f(x) = 1$, 
	\[
		1 - \epsilon \leq \frac{P(x)}{Q(x)} \leq 1 + \epsilon
	\]
	and hence 
	\[
		(1 - \epsilon)Q(x) \leq P(x) \leq (1 + \epsilon)Q(x).
	\]
\item If $f(x) = 0$,
	\[
		0 \leq \frac{P(x)}{Q(x)} \leq \epsilon
	\]
	and hence 
	\[
		0 \leq P(x) \leq \epsilon Q(x).
	\]
\end{itemize}
To approximate $f$ with bounded error $\epsilon$, all we need to be able to do is determine, up to some reasonable level of error, which of these cases is true. Write $P(x) = \sum_{S \subseteq [n]} \alpha_S X_S$ and $Q(x) = \sum_{S \subseteq [n]} \beta_S Y_S$, and let $\gamma = \sum_S \alpha_S + \beta_S$. We will use the notation $M \in P$ (resp. $M \in Q$) to denote the event that the monomial $M$ belongs to the polynomial $P$ (resp. $Q$).
\\
\\
Consider the following post-selection query algorithm:
\begin{itemize}
\item Choose $k$ monomials $M_1, M_2, \dots, M_k$ (independently, and with replacement) from $P$ or $Q$ with probabilities $p(X_S) = \alpha_S/\gamma$, $p(Y_S) = \beta_S/\gamma$.
\item If \emph{all $k$} monomials evaluate to 1, then:
	\begin{itemize}
	\item If all $k$ monomials belong to $P$, i.e $M_i \in P~ \forall i \in [k]$, return 1.
	\item Else if all $k$ monomials belong to $Q$, i.e. $M_i \in Q~ \forall i \in [k]$, return 0 with probability $r$.
	\item Else, return $\perp$.
	\end{itemize}
\item Else, return $\perp$.
\item We post-select on not seeing $\perp$. 
\end{itemize}
The query complexity of this algorithm is at most $kd$. To see correctness, consider the case $f(x)=1$. The probability of the algorithm returning 1, conditioned on not seeing $\perp$, is

\begin{eqnarray*}
\Pr[\text{see 1}|\text{not see}\perp] &=& \frac{\Pr[\text{see 1}]}{\Pr[\text{not see}\perp]} 
\\
&=& \frac{\prod_{i=1}^k \Pr[M_i=1]\cdot\Pr[M_i \in P|M_i=1]}
{\prod_{i=1}^k \Pr[M_i=1]\cdot\Pr[M_i \in P|M_i=1] + r \cdot \prod_{i=1}^k \Pr[M_i=1]\cdot\Pr[M_i \in Q|M_i=1]} 
\\
&=& \frac{\prod_{i=1}^k \Pr[M_i \in P|M_i=1]}{\prod_{i=1}^k \Pr[M_i \in P|M_i=1] + r \cdot \prod_{i=1}^k \Pr[M_i \in Q|M_i=1]} 
\\
&=& \frac{\frac{1}{\gamma^k}\left(\sum_{S: X_S=1}\alpha_S\right)^k}
{\frac{1}{\gamma^k}\left(\sum_{S: X_S=1}\alpha_S\right)^k + \frac{r}{\gamma^k}\left(\sum_{S: Y_S=1}\beta_S\right)^k}
\\
&=& \frac{ \left(\sum_{S}\alpha_S X_S\right)^k}
{\left( \sum_{S}\alpha_S X_S \right)^k + r \cdot\left(\sum_{S}\beta_S Y_S\right)^k} 
\\
&=& \frac{P(x)^k}{P(x)^k + r \cdot Q(x)^k} 
\\
\end{eqnarray*}
\\
But since $f(x)=1$, we have $(1-\epsilon)Q(x) \leq P(x) \leq (1+\epsilon)Q(x)$, and so 
\begin{eqnarray*}
\frac{P(x)^k}{P(x)^k + r \cdot Q(x)^k} &\geq& \frac{P(x)^k}{P(x)^k + \frac{r}{(1-\epsilon)^k}P(x)^k} \\
&=& \frac{(1-\epsilon)^k}{(1-\epsilon)^k + r}.
\end{eqnarray*}
\\
We can perform the equivalent calculation for the case $f(x) = 0$:
\begin{eqnarray*}
\Pr[\text{see 0}|\text{not see}\perp] &=& \frac{\Pr[\text{see 0}]}{\Pr[\text{not see}\perp]}
\\
&=& \frac{\prod_{i=1}^k \Pr[M_i=1]\cdot r \cdot \Pr[M_i \in Q|M_i=1]}
{\prod_{i=1}^k\Pr[M_i=1]\cdot\Pr[M_i \in P|M_i=1] +  r \cdot \prod_{i=1}^k\Pr[M_i=1]\cdot\Pr[M_i \in Q|M_i=1]} 
\\
&=& \frac{ r \cdot \prod_{i=1}^k \Pr[M_i \in P|M_i=1]}
{\prod_{i=1}^k\Pr[M_i \in P|M_i=1] +  r \cdot \prod_{i=1}^k\Pr[M_i \in Q|M_i=1]} 
\\
&=& \frac{\frac{r}{\gamma^k}\left(\sum_{S: Y_S=1}\beta_S\right)^k}
{\frac{1}{\gamma^k}\left(\sum_{S: X_S=1}\alpha_S\right)^k + \frac{r}{\gamma^k}\left(\sum_{S: Y_S=1}\beta_S\right)^k }
\\
&=& \frac{r\cdot\left(\sum_{S}\beta_S Y_S\right)^k}
{\left(\sum_{S}\alpha_S X_S\right)^k +  r \cdot \left(\sum_{S}\beta_S Y_S\right)^k} 
\\
&=& \frac{r\cdot Q(x)^k}{P(x)^k +  r \cdot Q(x)^k} 
\end{eqnarray*}
\\
In the case $f(x)=0$, $0 \leq P(x) \leq \epsilon Q(x)$, and so
\begin{eqnarray*}
\frac{ r \cdot Q(x)^k}{P(x)^k +  r \cdot Q(x)^k} &\geq& \frac{ r \cdot Q(x)^k}{\epsilon^k Q(x)^k +  r \cdot Q(x)^k} \\ 
&=& \frac{r}{\epsilon^k + r}.
\end{eqnarray*}
Choosing $r = \sqrt{\epsilon^k(1-\epsilon)^k}$ maximises the difference between $\frac{(1-\epsilon)^k}{(1-\epsilon)^k + r}$ and $1 - \frac{r}{\epsilon^k + r}$. With this choice of $r$, the probability of failure (i.e. that the algorithm returns 0 when $f(x) = 1$ and 1 when $f(x) = 0$) is
\[
	p_{\text{fail}} \leq  \frac{\sqrt{\epsilon^k}}{\sqrt{\epsilon^k} + \sqrt{(1-\epsilon)^k}}.
\]
Choosing $k=2$, we have 
\[
	p_{\text{fail}} \leq \frac{\epsilon}{\epsilon + (1-\epsilon)} = \epsilon.
\]
Therefore, by choosing $k=2$ and $r = \epsilon(1-\epsilon)$, we obtain a post-selected classical query algorithm that $\epsilon$-approximates $f$, and has degree at most $2\rdeg^+_\epsilon(f)$, implying that $\PR_\epsilon \leq 2\rdeg^+_\epsilon$. Note that the above algorithm still works for polynomials over the variables $\{x_i,\dots,x_N\} \cup \{(1-x_i),\dots,(1-x_N)\}$, since all we have to do to determine if a monomial $x_i,\dots,x_d(1-x_{i'}),\dots,(1-x_{d'})$ evaluates to 1 is check that all $x_i,\dots,x_d = 1$ and all $x_{i'},\dots,x_{d'} = 0$.  
\end{proof}
Following the first version of this paper, Ronald de Wolf pointed out that there is an alternative proof for Lemma \ref{lem:7} that uses the relationship between polynomials with positive coefficients and the model of query complexity in expectation:
\begin{proof}[Proof (Lemma \ref{lem:7})]
Consider an $n$-bit Boolean function $f$, and suppose that there exists a rational function $P/Q$ with degree $d=\rdeg_\epsilon(f)$ and positive coefficients that $\epsilon$-approximates $f$. Then $P$ and $Q$ are both degree-$d$ nonnegative literal polynomials such that $Q(x)>0$ and $|P(x)/Q(x) - f(x)| \leq \epsilon$ for all $x$. By Lemma \ref{lem:ldeg}, there exist classical $d$-query algorithms $A$ and $B$ whose \emph{expected} outputs equal $P(x)$ and $Q(x)$, respectively. Without loss of generality, we can assume that such algorithms only output either $0$ or some large fixed number $m$ (if the algorithms instead output a lower value $m'$, we can replace them with algorithms that output $m$ with probability $m'/m$, and $0$ otherwise, preserving the expected output value). By definition, we have $\Pr[A(x) = m] = P(x)/m$ and $\Pr[B(x) = m]=Q(x)/m$. Consider the following post-selected query algorithm $\mathcal{A}$:
\begin{itemize}
\item Run $B$ on input $x$. If $B(x)=0$, return $\perp$.
\item Otherwise, run $A$ on input $x$. If $A(x)=m$, return $1$, otherwise return $0$.
\item Post-select on not obtaining $\perp$.
\end{itemize}
Note that $\Pr[B(x)=m] = Q(x)$, which is non-zero everywhere since $Q(x)$ is non-zero everywhere by definition, and therefore post-selecting on the event that $B(x)=1$ is a valid move. The probability that the algorithm returns $1$ is 
\[
\Pr[\mathcal{A}(x) = 1 | \mathcal{A}(x) \neq \perp] = \frac{\Pr[A(x)=B(x)=m]}{\Pr[B(x)=m]} = \frac{\Pr[A(x)=m]}{\Pr[B(x)=m]} = \frac{P(x)/m}{Q(x)/m} = \frac{P(x)}{Q(x)},
\]
which is close to $f(x)$ by assumption.
\end{proof}

\section{Lower bound on the Majority function}\label{sec:lower_bound}
In this section we prove a $\Omega(N)$ lower bound on the bounded-error post-selected classical query complexity of the Majority function.
We consider the problem of estimating the Majority function up to error $\epsilon$, whose post-selected \emph{quantum} query complexity is $\Theta(\log (N / \log(1/\epsilon)) \log (1/\epsilon))$ (where the lower bound was proved by Sherstov \cite{sherstov2009intersection}, and the corresponding upper bound by Mahadev and de Wolf \cite{mahadev2015rational}). Here we prove a lower bound of $\Omega\left(N\left(1 - \sqrt{\frac{\epsilon}{1-\epsilon}}\right)\right)$, by using the relationship between post-selected classical query complexity and the degree of rational functions with positive coefficients. 

The main tool that we will use to prove the lower bound is the well-known symmetrization technique, due
to Minsky and Papert \cite{minsky2017perceptrons}:
\begin{lemma}\label{lem:minsky}[Minsky and Papert]
Let $f : \{0,1\}^N \rightarrow \{0,1\}$ be a symmetric (i.e. invariant under permutations of the bits of the input string) Boolean function on $N$ variables, and let $P: \{0,1\}^N \rightarrow \{0,1\}$ be a degree-$d$ polynomial that $\epsilon$-approximates $f$. Then there exists a univariate polynomial $p : \mathbb{R} \rightarrow \mathbb{R}$,
\[
p(k) =  \EE_{|X|=k}\left[p(x_1,\dots,x_N)\right]
\]
of degree at most $d$, such that $p(|x|) = P(x)$ for all $x \in \{0,1\}^N$, and hence $p$ $\epsilon$-approximates $f$.
\end{lemma}
Our first step is to explicitly write down the form of the univariate polynomials that correspond to multivariate polynomials in the $2N$ variables $\{x_0, \dots, x_N\} \cup \{(1-x_0), \dots, (1-x_N)\}$. 
\begin{claim}\label{lem:form}
Let $P : \{0,1\}^{2N} \rightarrow \{0,1\}$ be a degree-$d$ polynomial in the variables $\{x_0, \dots, x_N\} \cup \{(1-x_0), \dots, (1-x_N)\}$, which is necessarily of the form
\[
P(x) = \stsum a_{S,T} \prod_{i \in S} x_i \prod_{j \in T} (1-x_j).
\]
Then there exists a univariate polynomial $p$ of degree at most $d$, of the form
\[
p(k) = \stsum a_{S,T} \frac{(N-|S|-|T|)!}{N!} k(k-1)\dots(k-|S|+1)\cdot (N-k)(N-k-1)\dots(N-k-|T|+1)
\]
such that, if $P$ $\epsilon$-approximates a symmetric Boolean function $f$, then $p$ also $\epsilon$-approximates $f$, and depends only on the Hamming weight of the input string.
\end{claim}
\begin{proof}
In the usual way, we take the expectation of $P$ over permutations of the input bits, or equivalently, over all bit strings with the same Hamming weight as the input. That is, we consider
\begin{eqnarray*}
p(k) &=& \EE_{|x|=k} [P(x)] \\
&=& \EE_{|x|=k} \left[ \stsum a_{S,T} \prod_{i \in S} x_i \prod_{j \in T} (1-x_j) \right] \\
&=& \stsum a_{S,T} \EE_{|x|=k} \left[\prod_{i \in S} x_i \prod_{j \in T} (1-x_j) \right].
\end{eqnarray*}
The expectation over strings of Hamming weight $k$ can be written as
\begin{eqnarray*}
\EE_{|x|=k} \left[\prod_{i \in S} x_i \prod_{j \in T} (1-x_j) \right] &=& \left. {N-|S|-|T| \choose k-|S|} \middle/ {N \choose k} \right. \\
&=& \frac{(N-|S|-|T|)!}{(k-|S|)!(N-k-|T|)!} \cdot \frac{k!(N-k)!}{N!} \\
&=& \frac{(N-|S|-|T|)!}{N!} \cdot \frac{k!}{(k-|S|)!} \cdot \frac{(N-k)!}{(N-k-|T|)!} \\
\end{eqnarray*}\vspace{-0.75cm}
\[= \frac{(N-|S|-|T|)!}{N!} k(k-1)\dots(k-|S|+1)\cdot (N-k)(N-k-1)\dots(N-k-|T|+1) \]
which is a degree $(|S|+|T|)$ polynomial in $k$. Hence, $p(k)$ is a degree-$d$ polynomial in $k$.
\end{proof}
We remark that the coefficients $a_{S,T}$ are preserved in the transformation from $P$ to the univariate polynomial $p$. Therefore if $P$ is a polynomial with nonnegative coefficients, then $p$ has the property that all of the coefficients $a_{S,T}$ are nonnegative -- a fact that will be useful later on. 
\\
\\
Finally, we check that the symmetrization technique can still be used when working with rational functions.
\begin{lemma}\label{lem:rational_symm}
Let $R(x) = P(x)/Q(x)$ be a degree $d$ rational function that $\epsilon$-approximates a symmetric $N$-bit Boolean function $f : \{0,1\}^N \rightarrow \{0,1\}$. Then there exist univariate polynomials $p$ and $q$ of degrees at most $d$ such that 
\[
\left| f(x) - \frac{p(|x|)}{q(|x|)} \right| \leq \epsilon,
\]
where $p(k) = \EE_{|x|=k}[P(x)]$ and $q(k) = \EE_{|x|=k}[Q(x)]$.
\end{lemma}
\begin{proof}
Since $P(x)/Q(x)$ $\epsilon$-approximates $f(x)$, we have
\[
f(x) - \epsilon \leq P(x)/Q(x) \leq f(x) + \epsilon 
\]
and since $Q(x)$ is positive for all $x$,
\[
Q(x)(f(x) - \epsilon) \leq P(x) \leq Q(x)(f(x) + \epsilon).
\]
We take the expectation over permutations of the input bits (which is equivalent to taking the expectation over strings with identical Hamming weights). By Lemma \ref{lem:minsky}, there exist univariate polynomials $p$ and $q$ such that $\E_\sigma[P(\sigma(x))] = p(|x|)$ and $\E_\sigma[Q(\sigma(x))] = q(|x|)$, where the expectation is taken over permutations $\sigma \in S_N$. Then since $f$ is invariant under permutations of the input bits, 
\begin{eqnarray*}
\E_\sigma[Q(\sigma(x))(f(\sigma(x)) - \epsilon)] \leq& \E_\sigma[P(\sigma(x))] &\leq \E_\sigma[Q(\sigma(x))(f(\sigma(x)) + \epsilon)] \\
f(x)\E_\sigma[Q(\sigma(x))] - \epsilon\E_\sigma[Q(\sigma(x))] \leq& \E_\sigma[P(\sigma(x))] &\leq f(x)\E_\sigma[Q(\sigma(x))] + \epsilon\E_\sigma[Q(\sigma(x))] \\
q(|x|)f(x) - q(|x|)\epsilon \leq& p(|x|) &\leq q(|x|)f(x) + q(|x|)\epsilon \\
f(x) - \epsilon \leq& p(|x|)/q(|x|) &\leq f(x) + \epsilon \\
\end{eqnarray*}
which proves the result. 
\end{proof}
Note that if $P$ and $Q$ are polynomials with nonnegative coefficients in the variables $\{x_0, \dots, x_N\} \cup \{(1-x_0), \dots, (1-x_N)\}$ (i.e. precisely those polynomials that correspond to the acceptance probabilities of post-selected classical query algorithms), then $p$ and $q$ are both of the form described in Claim \ref{lem:form}, where each coefficient $a_{S,T}$ is nonnegative. 
\\
\\
We are now ready to prove the lower bound.
\begin{theorem}\label{theo:bound}
$\PRe(MAJ_N) \geq \frac{N}{2}\left(1 - \sqrt{\frac{\epsilon}{1-\epsilon}}\left(1 + o(1) \right)\right)$.
\end{theorem}
\begin{proof}
To prove the lower bound, we prove a lower bound on the degree of rational functions with positive coefficients that approximate the Majority function, and then use Theorem \ref{theo:main} to obtain a lower bound on the post-selected query complexity. 
To begin, consider a degree-$d$ rational function $P(x)/Q(x)$, with nonnegative coefficients, that approximates the Majority function up to error $\epsilon$. Then by Lemma \ref{lem:rational_symm}, there exist univariate polynomials $p$ and $q$ such that $p(|x|)/q(|x|)$ $\epsilon$-approximates the Majority function. Since $MAJ(x) = 1$ when $|x| > N/2$ and $MAJ(x) = 0$ when $|x| \leq N/2$, we have the inequalities
\begin{equation}\label{eq:ineq1}
(1-\epsilon) q\left(\n + 1\right) \leq p\left(\n+1\right) \leq (1+\epsilon) q\left(\n+1\right)
\end{equation}
and 
\begin{equation}\label{eq:ineq2}
0 \leq p\left(\n\right) \leq \epsilon q\left(\n\right).
\end{equation}
From Lemma \ref{lem:form}, we can write  
\[
q\left(\n\right) = \stsum \gamma_{S,T} b_{S,T} \n \left( \n-1 \right) \dots\ \left( \n-|S|+1 \right) \cdot\ \n\left(\n-1\right)\dots\ \left(\n-|T|+1\right) 
\]
and
\begin{eqnarray*}
\hspace{-2cm}
q\left(\n + 1\right) =& \stsum \gamma_{S,T} b_{S,T} \left(\n+1\right)\n\left(\n-1\right)\dots\ \left(\n-|S|+2\right)\cdot \left(\n-1\right)\dots\ \left(\n-|T|+1\right)\left(\n-|T|\right). \\
\end{eqnarray*}
Defining $\Delta_{S,T} := \n\left(\n-1\right)\dots\ \left(\n-|S|+1\right)\cdot \n\left(\n-1\right)\dots\ \left(\n-|T|+1\right)$, these become
\[
q\left(\n\right) = \stsum \gamma_{S,T} b_{S,T} \Delta_{S,T}
\]
and
\[
q\left(\n+1\right) = \stsum \gamma_{S,T} b_{S,T} \Delta_{S,T} \cdot \frac{\left(\n+1\right)\left(\n - |T|\right)}{\n \left(\n - |S| + 1\right)}.
\]
Our approach will be to bound the value of the term on the right for different values of $d$.
Since $|S| + |T| \leq d$, we have the following lower bound
\begin{eqnarray*}
\frac{\left(\n+1\right)\left(\n - |T|\right)}{\n \left(\n - |S| + 1\right)} 
&\geq& \frac{\left(\n+1\right)\left(\n - d\right)}{\n \left(\n + 1\right)} \\
&=& \frac{N - 2d}{N}.
\end{eqnarray*}

Since $\gamma_{S,T}, b_{S,T}, \Delta_{S,T} \geq 0$ for all choices of $S$ and $T$, this implies that 
\begin{equation}
q\left(\n\right) \leq \frac{N}{N-2d} \cdot q\left(\n + 1\right).
\end{equation}
Combining this with inequalities (\ref{eq:ineq1}) and (\ref{eq:ineq2}), we have
\[
p\left(\n\right) \leq \epsilon \cdot q\left(\n\right) \leq \frac{\epsilon N}{N-2d} \cdot q\left(\n + 1\right) \leq \frac{\epsilon N}{(1-\epsilon)(N-2d)} \cdot p\left(\n+1\right)
\]
and hence
\begin{equation}\label{ineq:contradiction}
p\left(\n\right) \leq \frac{\epsilon N}{(1-\epsilon)(N-2d)} \cdot p\left(\n+1\right).
\end{equation}
Similarly, we can write
\[
p\left(\n\right) = \stsum \gamma_{S,T} a_{S,T} \Delta{S,T}
\]
and 
\[
 \frac{\epsilon}{(1-\epsilon)(N-2d)} \cdot p\left(\n + 1\right) = \stsum \gamma_{S,T} a_{S,T} \Delta{S,T} \cdot  \frac{\epsilon}{(1-\epsilon)(N-2d)} \frac{\left(\n+1\right)\left(\n - |T|\right)}{\n \left(\n - |S| + 1\right)}.
\]
This time we want to \emph{upper bound} the value of the term on the right. Again using the fact that $|S| + |T| \leq d$, we have
\begin{eqnarray*}
 \frac{\epsilon N}{(1-\epsilon)(N-2d)} \frac{ \left(\n+1\right)\left(\n - |T|\right)}{\n \left(\n - |S| + 1\right)} 
&\leq& \frac{\epsilon N \left(\n + 1\right) \n}{(1-\epsilon)(N-2d)\n \left(\n -d + 1\right)} \\
&=& \frac{\epsilon N\left(\n+1\right)}{(1-\epsilon)(N-2d)\left(\n -d + 1\right)}.
\end{eqnarray*}
We consider the values of $d$ for which this quantity is strictly smaller than 1:
\begin{eqnarray*}
\frac{\epsilon N\left(\n+1\right)}{(1-\epsilon)(N-2d)\left(\n -d + 1\right)} &<& 1 \\
\frac{\epsilon}{(1-\epsilon)} N\left(\n+1\right)&<& (N-2d)\left(\n -d + 1\right) \\
0 &<& 2d^2 - 2(N+1)d + \left(1 - \frac{\epsilon}{(1-\epsilon)}\right)N(\n + 1).
\end{eqnarray*}
This inequality is satisfied when 
\[
2d < (N+1) - \sqrt{\frac{\epsilon}{1-\epsilon}(N^2+2N) + 1}.
\]
\\
Once again, using the fact that $\gamma_{S,T}, b_{S,T}, \Delta_{S,T} \geq 0$ for all choices of $S$ and $T$, we find that if $d$ satisfies the above constraint, then inequality (\ref{ineq:contradiction}) cannot be satisfied. Therefore, to $\epsilon$-approximate the Majority function, any rational function with positive coefficients must have degree at least $\frac{1}{2} \left((N+1) - \sqrt{\frac{\epsilon}{1-\epsilon}(N^2+2N) + 1}\right)$. In particular, if we take $\epsilon = 1/3$, then the degree $d$ must satisfy 
\[
d \geq \frac{N}{8}
\]
and so $\PRR(MAJ) = \Omega(N)$. More generally, we have $\PRe(MAJ_N) \geq \frac{N}{2}\left(1 - \sqrt{\frac{\epsilon}{1-\epsilon}}\left(1 + o(1) \right)\right)$.

\end{proof}

\noindent It turns out that this proof can be generalised in a straightforward manner to encompass \emph{any} symmetric function. In Appendix \ref{app:generalisation}, we prove the following:
\begin{restatable}{theorem}{theogeneralisation}\label{theo:generalisation}
Let $f : \{0,1\}^N \rightarrow \{0,1\}$ be any non-constant symmetric function on $N$ bits. Write $f_k = f(x)$ for $|x| = k$, define
\[
\Gamma(f) = \min\{|2k - N + 1| : f_k \neq f_{k+1}, 0 \leq k \leq N-1,
\]
and let $T$ be the value of $k$ for which this quantity is minimised. 
Then 
\[
\PR(f) \geq \frac{1}{8}\left(N - \Gamma(f)\right),
\]
and more generally for $\epsilon \in [0,1/2]$, 
\[
\PRe(f_T) \geq  \frac{1}{2}(N+1) \pm \sqrt{(N+1)^2 - 4\left(\frac{\epsilon}{1-\epsilon}\right)(N-T)(T+1)}.
\]
\end{restatable}
\noindent Intuitively, $\Gamma(f)$ is a measure of how close to the middle a `step change' (from 0 to 1, or vice versa) occurs. For instance, $\Gamma(OR) = N-1$ and $\Gamma(MAJ) = 1$.

\subsection{Post-selection algorithm for Majority}
Here we describe a classical post-selected query algorithm for computing Majority up to bounded error $1/3$, which uses $N/2 + 1$ queries to the input. This matches the lower bound up to a constant factor. Given an $N$-bit input $x$, the algorithm is as follows:
\begin{itemize}
\item Choose a subset of bits of size $\n + 1$.
\item If all bits are $1$, then return 1.
\item Else with probability $r$ return 0.
\item Else return $\perp$.
\end{itemize}
We post-select on not seeing $\perp$. To see that this algorithm works, consider an input such that $|x| = k$. Then the probability of the algorithm returning 1 is 
\begin{eqnarray*}
p_1 := \Pr[\text{see } 1 | \text{not see}\perp] &=& \frac{\left. {k \choose \n+1} \middle/ {N \choose \n+1} \right. }{\left. {k \choose \n+1} \middle/ {N \choose \n+1} \right. + \left(1- \left. {k \choose \n+1} \middle/ {N \choose \n+1} \right.\right)\cdot r} \\
&=& \frac{1}{1 - r + r\cdot \left. {N \choose \n+1} \middle/ {k \choose \n+1} \right.} \\
&\geq& \frac{1}{1 + r\cdot \left. {N \choose \n+1} \middle/ {k \choose \n+1} \right.} .
\end{eqnarray*}
When $f(x) = 1$, the hardest case to decide is when $k = \n + 1$, and therefore 
\begin{eqnarray*}
p_1 &\geq& \frac{1}{1 + r\cdot \left. {N \choose \n+1} \middle/ {\n +1 \choose \n+1} \right.}  \\
&=& \frac{1}{1 + r\cdot {N \choose \n+1} }. 
\end{eqnarray*}
Choosing $r = \frac{1}{2{N \choose \n+1}}$, we have $p_1 \geq 2/3$. 

In the case that $f(x)=0$, we have $k \leq \n$, and therefore $p_1 = 0$ always. Hence, this algorithm gives us bounded one-sided error of $1/3$. 
\\
\\
Since the error is one-sided, we can take $r$ to be arbitrarily small, taking the error arbitrarily close to zero without making any extra queries to the input. In this way, we obtain an algorithm for approximating the Majority function using at most $\n + 1$ queries for \emph{any} error $\epsilon \in (0,1/2)$. 
\\
\\
Conversely, the polynomial degree (both exact and approximate) of the Majority function is $N$~\cite{paturi1992degree}, and hence we can obtain at most a constant factor improvement over the decision tree complexity by using post-selection.

\section{Approximate counting using post-selection}\label{sec:counting}
It is possible to characterise $\PostBPP$ in terms of approximate counting problems (see e.g. \cite{o2016weakness}). This observation provides an alternative perspective on quantum supremacy results~\cite{harrow2017quantum}, as well as giving us a `post-selection free' way to define the complexity classes related to $\PostBPP$. Approximate counting was introduced by Sipser~\cite{sipser1983complexity} and Stockmeyer~\cite{stockmeyer1983complexity, stockmeyer1985approximation} as the problem of estimating the number of satisfying assignments to a given circuit or formula. In the query complexity setting, we can view this as the problem of estimating the Hamming weight $|x|$ of a bit-string $x$. In this context, the `weak' approximate counting problem is defined to be the problem of estimating $|x|$ up to a factor of 2, and `strong' approximate counting to be the problem of estimating $|x|$ up to a factor of $(1+1/p)$ for some `accuracy parameter' $p\geq1$. 

Here we show that there exist efficient post-selected classical query algorithms for both versions of the approximate counting problem. By Theorem \ref{theo:main}, this implies that there exist low-degree rational functions with positive coefficients for approximating the Hamming weight of a bit-string. 

\begin{theorem}
Given an $N$-bit string $x$ and an integer $p \geq 1$, there exists a post-selected classical query algorithm that outputs $\widetilde{|x|}$ such that
\[
\frac{1}{(1+1/p)} |x| \leq \widetilde{|x|} \leq (1+1/p)|x|
\]
with probability $\geq 1-\epsilon$, by making at most $O(p \log N \log(\log N/\epsilon))$ queries to the input.
\end{theorem}
\begin{proof}

Consider the following algorithm, $\mathcal{V}(A, x)$, where $A$ is a `guess' at the value of $|x|$:
\begin{itemize}
\item Choose $k$ input bits uniformly at random, with replacement.
\item If $x_i = 1$ for all chosen bits, return 1.
\item Else if $x_i = 0$ for all chosen bits, then with probability $r$, return 0.
\item Else return $\perp$.
\item Post-select on not seeing $\perp$.
\end{itemize}
Clearly, only k queries are required. The probability of the algorithm outputting 1 is
\begin{eqnarray*}
p_1 := \Pr[\mathcal{V}(A,x) = 1 | \mathcal{V}(A,x) \neq \perp] &=& \frac{\Pr[\mathcal{V}(A,x)=1]}{\Pr[\mathcal{V}(A,x) \neq \perp]} \\
&=& \frac{\left(|x|/N\right)^k}{\left(|x|/N\right)^k + (1-|x|/N)^kr} \\
&=& \frac{1}{1 + r\frac{N^k}{|x|^k}(1 - |x|/N)^k}.
\end{eqnarray*}
Our approach will be to guess a value $A = 2^i$ for some $0 \leq i \leq \lceil \log_2 N \rceil$. Then by choosing $r$ appropriately, the algorithm will output 1 with high probability if our guess is at least a factor of 2 below the true value of $|x|$, and will output 0 with high probability if our guess is larger than the true value of $|x|$ by a factor of 2. We will choose $r = 2A^k/N^k$, so that
\[
p_1 = \frac{1}{1 + \frac{2A^k}{|x|^k}\left(1 - \frac{|x|}{N} \right)^k}.
\]
Consider a guess $A$ for which $A < |x|/2$. In this case,
\[
p_1 > \frac{1}{1 + \frac{2}{2^k}\left( 1 - \frac{|x|}{N} \right)^k} \geq \frac{1}{1 + \frac{2}{2^k}\frac{(N-1)^k}{N^k}}.
\]
This value is exponentially (in $k$) close to 1.
Now consider a guess $A$ for which $A > 2|x|$. In this case,
\[
p_1 < \frac{1}{1 + 2\cdot2^k\left( 1 - \frac{|x|}{N} \right)^k}.
\]
Since $A > 2|x|$, we can assume that $|x| \leq N/2$ (since we won't guess a value of $A$ that is greater than $N$), and thus
\[
p_1 < \frac{1}{1 + \frac{2\cdot2^k}{2^k}} = \frac{1}{3}.
\]
By taking $k=2$, when $A \leq |x|/2$, $p_1 \geq 2/3$, and when $A \geq 2|x|$, $p_1 \leq 1/3$. (Interestingly, choosing a larger value of $k$ doesn't help us, since it can only increase the probability that the algorithm returns 1 when $A \leq |x|/2$, but not the probability that the algorithm returns 0 when $A \geq 2|x|$).
\\\\
As we double our guess $A$, we will at some point encounter a run of the algorithm that returns 0 with probability greater than $2/3$. At this point, we can stop our search, and output $A$ as our estimate for $|x|$. This estimate will be within a factor of $4$ of the true value of $|x|$, with high probability.

The overall algorithm will be:
\begin{enumerate}
\item For $i = 0$ to $\lceil\log_2 N\rceil$:
	\begin{enumerate}
	\item Set $A = \min\{2^i, N\}$, and run $\mathcal{V}(A,x)$ $O(\log(\log N/\epsilon))$ times.
	\item If the majority of runs return $0$, then we return $A$ as our guess at $|x|$.
	\item Otherwise, continue.
	\end{enumerate}
\item If $i = \lceil\log_2 N\rceil$, then return $N$ as our guess for $|x|$.
\end{enumerate}
By repeating the algorithm $\mathcal{V}$ $O(\log(\log N/\epsilon))$ times in step 1(a) and taking the majority, we amplify the acceptance/rejection probability of the algorithm: if $A \geq 2|x|$, the majority answer will be 0 with probability $\geq 1-\frac{\epsilon}{\log N}$, and if $A < |x|/2$, the majority answer will be 1 with probability $\geq 1-\frac{\epsilon}{\log N}$. For intermediate values of $A$, the algorithm might return 0 or 1 with probabilities close to 1/2. In this case, if the algorithm returns 1, we continue as normal. Else, if it returns 0, then we will be at most a factor of $1/2$ away from the true value of $|x|$. 

In the worst case, when $|x| > N/2$, then at most $m \leq \lceil\log_2 N\rceil$ rounds of the algorithm will be performed. In each round, the algorithm erroneously returns $1$ with probability at most $\epsilon / \log N$. Hence, by the union bound, the probability that the algorithm fails in at least one of these rounds is at most $m \epsilon / \log N = \epsilon$. Otherwise, the probability that the algorithm doesn't return $0$ in the first round that it guesses an $A \geq 2|x|$ is at most $\epsilon / \log N < \epsilon$. Thus, the algorithm fails with probability at most $\epsilon$. 

If the algorithm correctly halts after the first guess in which $A$ is larger than $2|x|$, then our estimate will be at most $4|x|$. 
Hence we obtain an estimate $\widetilde{|x|}$ of $|x|$ such that
\[
\frac{|x|}{2} \leq \widetilde{|x|} \leq 4|x|
\]
with probability at least $1-\epsilon$, making at most $O(\log N \log(\log N/\epsilon)$ queries to the input. We can improve the accuracy so that the estimate is within a factor of 2 by increasing the number of queries by a constant factor. Hence, this algorithm gives a solution to the `weak' approximate counting problem.

We can use Stockmeyer's trick~\cite{stockmeyer1985approximation} to obtain an algorithm that can produce an estimate of $|x|$ such that
\[
\frac{1}{(1 + 1/p)} |x| \leq \widetilde{|x|} \leq (1 + 1/p)|x|
\]
for arbitrary $p \geq 1$ by repeating the above algorithm $O(p)$ times. This is a solution to the `strong' approximate counting problem, and we can obtain such a solution using post-selection by making at most $O(p \log N \log(\log N/\epsilon)$ queries to the input. 
\end{proof}

We remark that this algorithm is almost optimal, since any algorithm that makes substantially fewer queries to the input would be able to violate the lower bound from Theorem \ref{theo:bound}. In particular, by choosing $p \geq \frac{1}{\sqrt{1+2/N}-1} \approx N$, and $\epsilon = 1/3$ above, we obtain an algorithm for approximating the Majority function up to accuracy $1/3$. Hence, any sub-linear dependence on $p$ is ruled out by the Majority lower bound.

\appendix 
\section{Post-selection algorithms with post-selection sub-routines}\label{app:nested}
In this section, we show that if one is allowed to use a post-selection algorithm as a sub-routine inside another post-selection algorithm, then all Boolean functions can be computed up to bounded error using only a single query.

The algorithm works as follows: Guess a bit-string $y$, post-select on it being the same as the input $x$ (using a single query), and then return $f(y)$. We first construct a verifier $V(x,y)$, which will accept with certainty if $y = x$, and reject with a probability arbitrarily close to 1 when $x \neq y$. It works as follows:
\begin{itemize}
\item Choose an index $i \in [n]$ uniformly at random.
\item If $x_i \neq y_i$, return 0.
\item Else, return 1 with probability $1/K$, for some $K$ to be determined later.
\item Else, return $\perp$.
\item Post-select on not seeing $\perp$.
\end{itemize}
This procedure requires a single query to $x$. To see correctness, first consider the case in which $y=x$. Then, the probability of seeing 1 is
\begin{eqnarray*}
\Pr[\text{see 1} | \text{not see} \perp] &=& \frac{\Pr[\text{see 1}]}{\Pr[\text{not see} \perp]} \\
&=& \frac{\Pr[x_i = y_i] \cdot \frac{1}{K} }{\Pr[x_i \neq y_i] + \Pr[x_i = y_i] \cdot \frac{1}{K}} \\
&=& 1.
\end{eqnarray*}
Now suppose that $y \neq x$. The hardest case to distinguish is when $y$ differs from $x$ on only a single bit:
\begin{eqnarray*}
\Pr[\text{see 0} | \text{not see} \perp] &=& \frac{\Pr[\text{see 0}]}{\Pr[\text{not see} \perp]} \\
&=& \frac{\Pr[x_i \neq y_i] }{\Pr[x_i \neq y_i] + \Pr[x_i = y_i] \cdot \frac{1}{K}} \\
&=& \frac{ \frac{1}{n} }{\frac{1}{n} + \left(1 - \frac{1}{n}\right) \cdot \frac{1}{K}} \\
&\geq& \frac{ \frac{1}{n} }{\frac{1}{n} +  \frac{1}{K}} \\
&=& \frac{ 1 }{ 1 +  \frac{n}{K}}.
\end{eqnarray*}
We can choose $K$ to be arbitrarily large in order to make this probability as close to 1 as needed. 

Now that we have our verifier, we can use it to find a $y$ such that $y = x$. The procedure is similar:
\begin{itemize}
\item Choose a $y$ uniformly at random. 
\item If $V(x,y) = 1$, return $f(y)$.
\item Else, return $\perp$.
\item Post-select on not seeing $\perp$.
\end{itemize}
Since each use of $V(x,y)$ requires a single query to $x$, this procedure again only requires a single query to the input. Let $A$ be the number of inputs in $f^{-1}(1)$, $N := 2^n$, and define $\epsilon := 1 - \frac{1}{1 + \frac{n}{K}}$. 

Suppose that $f(x) = 1$. The probability of seeing $1$ is the probability that we choose a $y \in f^{-1}(1)$ multiplied by the probability that the verifier accepts such a bit-string. If $x=y$, this happens with certainty, and if $x\neq y$, this happens with probability at most $\epsilon$. Thus,
\[
\frac{1}{N} \leq p_1 \leq \frac{A}{N}\left(\frac{1}{A} + \frac{\epsilon(A-1)}{A} \right) = \frac{1 + \epsilon A - \epsilon}{N}.
\]
On the other hand, the probability of seeing the `wrong' answer, 0, is given by the probability of choosing a $y \in f^{-1}(0)$ multiplied by the probability that the verifier accepts, which happens with probability at most $\epsilon$ (since $x\neq y$ in the case that $f(x)\neq f(y)$):
\[
p_0 \leq \frac{\epsilon(N-A)}{N}.
\]
Finally, 
\begin{eqnarray*}
\Pr[\text{see 1}|\text{not see}\perp] &=& \frac{\Pr[\text{see 1}]}{\Pr[\text{not see}\perp]} \\
&=& \frac{p_1}{p_1 + p_0} \\
&\geq& \frac{\frac{1}{N}}{\frac{1 + \epsilon A - \epsilon}{N} +  \frac{\epsilon(N-A)}{N}} \\
&=& \frac{1}{1 + \epsilon A - \epsilon +  \epsilon N - \epsilon A} \\
&=& \frac{1}{1 - \epsilon +  \epsilon N } \\
&\geq& \frac{1}{1 +  \epsilon N }.
\end{eqnarray*}
By choosing $\epsilon = \frac{1}{2N}$, and hence $K = 2nN$, we can ensure that this probability is greater than $2/3$. An identical argument follows in the case that $f(x) = 0$. 

Note that we can't `flatten' this algorithm to obtain one that produces the same result, but only performs one post-selection step.

\section{Generalisation of the lower bound to arbitrary symmetric functions}\label{app:generalisation}
Here we prove the more general version of the lower bound from Theorem \ref{theo:generalisation}, restated here for convenience:
\theogeneralisation*
\begin{proof}
Since $f_T \neq f_{T+1}$, there must be a step-change in the value of $f(x)$ at the point $|x| = T$. Without loss of generality, we can assume that this change is such that $f(T) = 0$ and $f(T+1) = 1$. Then, assuming that there exists some (univariate) rational function $p(x)/q(x)$ that $\epsilon$-approximates $f$, we have
\[
q\left(T\right) = \stsum \gamma_{S,T} b_{S,T} \Delta_{S,T}
\]
and
\[
q\left(T+1\right) = \stsum \gamma_{S,T} b_{S,T} \Delta_{S,T} \cdot \frac{\left(T+1\right)\left(N - T - |T|\right)}{(N-T) \left(T - |S| + 1\right)}.
\]
In the following, we can assume that $|T| \leq N-T$ and $|S| \leq T+1$ (since any monomial with $|S|$ or $|T|$ larger would necessarily equal zero and not contribute to the polynomials considered above). 
We can lower bound the term on the right using the fact that $|S| + |T| \leq d$:
\[
\frac{\left(T+1\right)\left(N - T - |T|\right)}{(N-T) \left(T - |S| + 1\right)} \geq
\frac{\left(T+1\right)\left(N - T - d\right)}{(N-T) \left(T + 1\right)} = 
\frac{\left(N - T - d\right)}{(N-T) }.
\]
Hence, 
\[
q(T) \leq \frac{N-T}{N-T-d} \cdot q(T+1)
\]
and so 
\begin{equation}\label{eq:first_ineq}
p(T) \leq \epsilon q(T) \leq \frac{\epsilon(N-T)}{N-T-d} q(T+1) \leq \frac{\epsilon(N-T)}{(1-\epsilon)(N-T-d)}p(T+1).
\end{equation}
The next step is to upper bound the value of $\frac{\epsilon(N-T)}{(1-\epsilon)(N-T-d)}\cdot \frac{(T+1)(N - T - |T|)}{(N-T) (T - |S| + 1)}$. Using $|S| + |T| \leq d$, we have
\[
\frac{\epsilon(N-T)}{(1-\epsilon)(N-T-d)}\cdot \frac{(T+1)(N - T - |T|)}{(N-T) (T - |S| + 1)} \leq 
\frac{\epsilon(N-T)}{(1-\epsilon)(N-T-d)}\cdot \frac{(T+1)(N - T)}{(N-T) (T - d + 1)}
\]
\[
= \frac{\epsilon(N-T)}{(1-\epsilon)(N-T-d)}\cdot \frac{(T+1)}{(T - d + 1)}.
\]
We consider the values of $d$ for which this quantity is strictly smaller than 1:
\begin{eqnarray*}
\frac{\epsilon(N-T)}{(1-\epsilon)(N-T-d)}\cdot \frac{(T+1)}{(T - d + 1)} &<& 1 \\
\frac{\epsilon(N-T)(T+1)}{(1-\epsilon)}  &<& (N-T)(T+1) - (T+1)d - (N-T)d + d^2 \\
0 &<& d^2 - (N+1)d + (N-T)(T+1)\left(1 - \frac{\epsilon}{1-\epsilon}\right)   \\
\end{eqnarray*}
Writing $\alpha := \left(1 - \frac{\epsilon}{1-\epsilon}\right)$, the roots of the inequality are given by:
\begin{eqnarray*}
2d &=& (N+1) \pm \sqrt{(N+1)^2 - 4\alpha(N-T)(T+1)} \\
\end{eqnarray*}
We can obtain a lower bound by considering the value of $(N+1) - \sqrt{(N+1)^2 - 4\alpha(N-T)(T+1)}$. In this case, we have the lower bound
\[
2d \geq (N+1) - \sqrt{(N+1)^2 - 4\alpha(T+1)(N+1) + 4\alpha(T+1)^2}.
\]
If this bound is violated then inequality (\ref{eq:first_ineq}) cannot be satisfied, and hence no rational approximation can exist with degree smaller than this lower bound.
\\
\\
In particular, if we take $\epsilon = 1/3$, we have
\begin{eqnarray*}
2d &\geq& N + 1 - \sqrt{(N+1)^2 - 2(N-T)(T+1)} \\
&=& N + 1 - \sqrt{(N-T)^2 + (T+1)^2}.
\end{eqnarray*}
It is straightforward to check that a suitable lower bound is 
\begin{eqnarray*}
d &\geq& \frac{N}{8} - \frac{1}{4}\left| \n - T\right| = \frac{1}{8}(N - \Gamma(f)).
\end{eqnarray*}
\end{proof}
As we pointed out in Section \ref{sec:seps}, this result is quite similar to a result of Paturi~\cite{paturi1992degree}, and characterises the rational degree (and hence the post-selected classical query complexity) of symmetric functions.

\section*{Acknowledgements}

I would like to thank Ashley Montanaro for helpful discussions and careful comments on this paper, and Ronald de Wolf for helpful comments and for pointing out the results in reference \cite{kaniewski2015query}, as well as the alternative proof for Lemma \ref{lem:7}. This work was supported by the EPSRC. No new data were created in this study.

\bibliography{references_post_selection.bib}

\end{document}